\documentclass{llncs}
\usepackage{llncsdoc}
\usepackage[english]{babel}
\usepackage{amsmath}
\usepackage{amssymb,amsfonts,textcomp}
\usepackage{multicol}
\usepackage{array}
\usepackage{hhline}
\usepackage{graphicx} 
\begin{document}

\title{Enumeration Based Search Algorithm For Finding A Regular Bi-partite Graph Of Maximum Attainable Girth For Specified Degree And Number Of Vertices}

\author{Vivek S Nittoor \and Reiji Suda}

\institute{The University Of Tokyo}

\maketitle
\begin{abstract}
We introduce a search problem for finding a regular bi-partite graph of maximum attainable girth for specified degree and number of vertices, by restricting the search space using a series of mathematically rigourous arguments from $[1]$ and $[2]$. The goal of this paper is to derive the enumeration search algorithm for finding a girth maximum $(m, r)$ BTU, which is notation for regular partite graph that has been introduced in $[1]$, using the optimal partition results from $[2]$ as a starting point, and also understand the structure of the search space and the computational complexity of the algorithm.
\end{abstract}

\section{Introduction}
The goal of this paper is to build upon the results developed in $[1]$ and $[2]$ and and propose a search algorithm in order to find a girth maximum Balanced Tanner Unit (BTU).  The search algorithm is based upon  enumeration of all elements of a Symmetric Permutation Tree described in $[1]$ with a fixed node at depth $1$. The theoretical background behind BTUs  has been introduced and explained in detail in $[1]$ and $[2]$.

\subsection{Definitions}
We review definitions from  $[1]$ and  $[2]$ .

\begin{definition} $(m, r)$ BTU \\
A $(m,r)$ Balanced Tanner Unit (BTU) is a regular bi-partite graph that can be represented by a $m \times m$ square matrix with $r$ non-zero elements in each of its rows and columns. Every $(m,r)$ BTU has a bipartite graph representation and an equivalent matrix representation.
\end{definition}

\begin{definition} {Girth maximum $(m,r)$ BTU} \\
{A labeled $(m,r)$  BTU  $A$ is girth maximum if there does not exist another labeled $(m,r)$  BTU $B$ with girth greater than that of $A$ .}
\end{definition}
\begin{definition}{Symmetric Permutation Tree and its properties} \\
{A  $m$  Symmetric permutation tree  $S_{\mathit{PT}}\{m\}$  is defined as a labeled tree with the following properties:}
\begin{itemize}
\item {$S_{\mathit{PT}}\{m\}$ has a single root node labeled  $0$.}
\item {$S_{\mathit{PT}}\{m\}$ has $m$  nodes at \ depth  $1$ from the root node.}
\item {$S_{\mathit{PT}}\{m\}$ has nodes at depths ranging from  $1$ to $m$ ,
with each node having a labels chosen from  $\{1,2,\ldots ,m\}$ . The
root node  $0$ has  $m$  successor nodes. Each node at \ depth  $1$ has $m-1$  successor nodes at depth  $2$ . Each \ node at depth  $i$  has 
$m\text{--}i+1$  successor nodes at depth $i+1$  . Each node at depth 
$m\text{--}1$  has  $1$  successor node at depth $m$  .}
\item {No successor node in  $S_{\mathit{PT}}\{m\}$ has the same node label as any of its ancestor nodes.} {No two successor nodes that share a common parent node have the same label.}
\item {The sequence of nodes in the path traversal from the node at depth  $1$ to the leaf node at depth  $m$  in  $S_{\mathit{PT}}\{m\}$  represents
the permutation represented by the leaf node.}
\item {$S_{\mathit{PT}}\{m\}$ has  $m!$  leaf nodes each of which represent an element of the symmetric group of degree  $m$ denoted by  $S_{m}$.}
\end{itemize}
\end{definition}
\begin{definition}{$\Phi (\beta _{1},\beta _{2},\ldots ,\beta _{r - 1})$ where $\beta _{i}\in P_{2}(m)$ for  $1\le i\le r-1$} \\
 $\Phi (\beta _{1},\beta _{2},\ldots ,\beta _{r\text{--}1})$ refers to the family of all labeled $(m,r)$ BTUs with compatible permutations  $p_{1,}p_{2,}\ldots ,p_{r}\in S_{m};p_{i}\notin C(p_{1},p_{2},\ldots ,p_{i-1})$ for  $1<i\le r$ that occur in the same order on a complete $m$ symmetric permutation tree ,  $x_{1,1}<x_{2,1}<\ldots <x_{r,1}$ where  $p_{j}=(x_{j,1}x_{j,2}\ldots x_{j,m});1\le j\le r$ , such that  $\beta _{i-1}$  is the partition between permutations $p_{i\text{--}1}$  and  $p_{i}$  for all integer values of  $i$  given by  $1<i\le r$ .
\end{definition}
\begin{definition}{Optimal partition parameters for girth maximum $(m, r)$ BTU} \\
{ $\beta _{1},\beta _{2},\ldots ,\beta _{r-1}\in P_{2}(b\ast k^{r-1})$ {refer to optimal partitions derived in} $[2]$ {such that there exists a girth maximum} $(m,r)$ { BTU \ in } $\Phi (\beta _{1},\beta _{2},\ldots ,\beta _{r-1})$ {,where } $\beta _{i}$ refers to $\sum _{j=1}^{r-1-i}b\ast k^{i}=b\ast k^{r-1}$ for $1\le i\le r-1$. Thus, $\beta _{1},\beta _{2},\ldots ,\beta _{r-1}\in P_{2}(b\ast k^{r-1})$ {are } $\sum _{j=1}^{k^{r-2}}b\ast k=b\ast k^{r-1}${, } $\sum _{j=1}^{k^{r-3}}b\ast k^{2}=b\ast k^{r-1}${, } $\ldots ${, } $\sum _{j=1}^{k}b\ast k^{r-2}=b\ast k^{r-1}$ {and } $\sum _{j=1}^{1}b\ast k^{r-1}=b\ast k^{r-1}$ {respectively.}}
\end{definition}

\begin{definition} Compatible Permutations \\
Two permutations on a set of  $s$  elements represented by 
$(x_{1}x_{2}\ldots x_{s});x_{p}\neq x_{q}$
$\forall p\neq q;1\le p\le
s;1\le q\le s;p,q\in \mathbb{N}$  where  $1\le x_{i}\le s;i\in
\mathbb{N};1\le i\le s$ and  $(y_{1}y_{2}\ldots y_{s});y_{p}\neq y_{q}$
 $\forall p\neq q;1\le p\le s;1\le q\le s;p,q\in \mathbb{N}$  where 
$1\le y_{i}\le m;i\in \mathbb{N};1\le i\le s$ are compatible if and
only if  $x_{i}\neq y_{i}\forall i\in \mathbb{N};1\le i\le s$.
\end{definition}

\section{Parameters For BTU Search}

\subsection {Assumptions for a $(m, r)$ BTU}
\begin{enumerate}
\item We assume that $r < m/2$.
\item We assume that $m$ is a composite number and contains a power of $r - 1$ in its prime factorization.
\end{enumerate}

We present an algorithm from $[2]$  for generating optimal partitions for girth maximum $(m,r)$ BTU. 
\subsection {Algorithm from $[2]$ to generate optimal partitions for a given value of  $k$  and $r$ }
{The following algorithm generates optimal partitions  $\beta _{1},\beta _{2},\ldots ,\beta _{r - 1}\in P_{2}(k^{r-1})$ for a given value of  $k$  and $r$ such that the girth maximum $(k^{r-1},r)$  BTU lies in $\Phi (\beta _{1},\beta _{2},\ldots ,\beta _{r - 1})$ .} \\
 $z=k$ ; \\
for( $i=1;i\le r-1;i$ ++) \{ \\
 $\beta _{i}$ refers to partition $\sum _{j=1}^{1}z \in P_2(z)$ ; \\
for( $l=1;l<i;l$ ++) \{ \\
\  $k\ast \beta _{l}$ ; //scale partition $\beta _{l}$ by  $k$ \\
\} \\
 $z=k\ast z$ ; \\
\} 

\subsection{Calculation of $b, k \in \mathbb{N}$ from $m,r \in \mathbb{N}$ such that  $m >  r$}
Given $m,r \in \mathbb{N}$ such that  $m >  r$,\ $k \in \mathbb{N}$ is the solution of the equation $m=b * k^{r-1}$ such that  $b \in \mathbb{N}$ is minimized and is denoted by the function $k=f(m,r)$ where $f:\{\mathbb{N}\cup \{0\}\}^{2}\to \mathbb{N}\cup \{0\}$ . If $m$ does not contain a power of $r - 1$ in its prime factorization, then the routine gives us $k = 1$ and $b = m$ which is not useful for our Enumeration Based Search Algorithm.We calculate $k$ such that $m=b * k^{r-1}$ such that  $b \in \mathbb{N}$ is minimized as follows. \\
$b=1; $ \; 
\\if( $(m/b)^{1/r-1}\in \mathbb{N}$) \{ 
\\{ $k=(m/b)^{1/r-1}$; \\ 
\} 
\\{else \{}
\\{while( $(m/b)^{1/r-1}\notin \mathbb{N}$) \{}
\\$b++$;
\\ \}
\\{ $k=(m/b)^{1/r-1}$ \ ;
\\ \} 

\begin{definition}{Family of labeled BTUs $X(m, r)$} \\
$X(m, r)$ is a family of labeled BTUs each of which represented by set of compatible permutations ${p_{1} , p_{2} , \ldots, p_{r}}$ such that
\begin{enumerate}
\item $\beta _{i}${ refers to the optimal partition given by} $\sum _{j=1}^{r-1-i}b\ast k^{i}=b\ast k^{r-1}$ between permutations $p_{i}$  and  $p_{i + 1}$  {for } $1\le i\le r-1$ with $k \in \mathbb{N}$ is the solution of the equation $m=b * k^{r-1}$ such that  $b \in \mathbb{N}$ is minimized.
\item $p_{r - 1} =I_{m}$, where $I_m$ is the identity permutation on a set of $m$ elements.
\end{enumerate}
\end{definition}

\begin{theorem}
\item Every labeled  $(m,r)$  BTU in  $\Phi (\beta _{1},\beta_{2},\ldots ,\beta _{r\text{--}1})$ is isomorphic to some labeled $(m,r)$  BTU in  $X(m, r)$ , where  $\beta _{i}${ refers to the optimal partitions given by} $\sum _{j=1}^{r-1-i}b\ast k^{i}=b\ast k^{r-1}$ {for } $1\le i\le r-1$ and $p_{r - 1} =I_{m}$, where $I_m$ is the identity permutation on a set of $m$ elements.
\end{theorem}
\begin{proof}  The proof has been provided in the Appendix.
\end{proof}

\subsection {$C_j$ as a permutation that corresponds to a circular permutation matrix}
\begin{definition}
Given $j\in \mathbb{N}$ such that  $1\le j<m$ , we define a permutation that corresponds to circular permutation matrix $C_{j}$ in the following manner .
Starting with an identity matrix  $I_{m}$, we move $j$ rows from the end and move it to the top in order to obtain the circular permutation matrix $C_{j}$ .
\end{definition}

\begin{theorem}
A labeled $(m, 2)$ BTU can be constructed by a set of two compatible permutations, identity permutation $I_m$  and the permutation that corresponds to the circular permutation matrix $C_j$ such that  $j\in \mathbb{N}$ such that  $1\le j<m$, have a partition of $\sum _{i=1}^{l}(m/l) = m$ if $j$ divides $m$ where  $l = min(k, m - k)$ and partition $(m) \in P_2(m)$ if $j$ and $m$ are relatively prime. 
\end{theorem}
\begin{proof}  The proof has been provided in the Appendix.
\end{proof}

\begin{corollary}
The permutation represented by circular permutation matrix $C_j$ on $m$ elements such that  $1\le j<m$ has a partition $(m) \in P_2(m)$ with the identity permutation $I_m$ if $j$ and $m$ are relatively prime. 
\end{corollary}

\begin{definition}{Family of labeled BTUs $Z(m, r)$}  \\
$Z(m, r) \subset X(m, r)$ is a family of labeled BTUs each of which represented by set of compatible permutations \{${p_{1} , p_{2} , \ldots, p_{r}} \in S_{m}$\} such that
\begin{enumerate}
\item $\beta _{i}$ refers to the optimal partition given by $\sum _{j=1}^{r-1-i}b\ast k^{i}=b\ast k^{r-1}$ between permutations $p_{i}$  and  $p_{i + 1}$ for $1\le i\le r-1$ with $k \in \mathbb{N}$ being the solution of the equation $m=b * k^{r-1}$ such that  $b \in \mathbb{N}$ is minimized.
\item $p_{r - 1} =I_{m}$, where $I_m$ is the identity permutation on a set of $m$ elements.
\item $p_{j}=k^{r -1 - j}\ast q_{j}$ for $1 \le j \le r - 2$, and $q_{j} \in S_{b \ast k^{j}}$.
\end{enumerate}
\end{definition}
\section {Progressive Reduction Of Search Space} 
We start from the defintion of the Optimal partition parameters for girth maximum $(m, r)$ BTU derived in $[2]$, and restrict the search space.  Without loss of generality, we progressively reduce the search space for a girth maximum $(m, r)$ BTU in the following steps.  $\beta _{i}${ refers to the optimal partitions given by} $\sum _{j=1}^{r-1-i}b\ast k^{i}=b\ast k^{r-1}$ {for } $1\le i\le r-1$ in the following steps. Table $1$ shows a progressive reduction in the size of the search space for number of choices for $p_{1}, p_{2}, p_{3} \in S_{m}$ for a $(m, 3)$ BTU with $m = k^2$. Best permutation $p_{i -2}$ that maximizes the girth for each $(b \ast k^{i -1}, i)$  BTU can be found by search.
\begin{enumerate}
\item Space of all labeled $(m,r)$ BTUs.
\item $\Phi (\beta _{1},\beta _{2},\ldots ,\beta _{r - 1})$.
\item $Z(m, r) \subset X(m, r)$.
\end{enumerate}
\begin{table}
\caption{Size Of the Search Space for $r=3$ with $m = k^2$}
\begin{tabular}{llllll}
\hline\noalign{\smallskip}
$Search Space$ & $p_{1}$ choices  & $p_{2}$ choices   & $p_{3}$ choices  \\
\noalign{\smallskip}
\hline
\noalign{\smallskip}
Labeled  $(m,3)$ BTUs & $m!$ & $ \approx {(m-1)!}$ &  $ \approx {(m-2)!}$\\
$\Phi (\beta _{1},\beta _{2},\beta _{3})$ & $1$ & $1$ & ${(k-1)!}^{k -1}$   \\
$Z(m,3)$ & $1$ & $1$ & $(k-1)!$ \\
\hline
\end{tabular}
\end{table}

\section {Scaling}
{The validity of restricting the form of $p_{i-2}\in S_{b\ast k^{i-1}}$ to that of scaled versions of \  $q_{i-2}\in S_{b\ast k^{i-2}}$  at each stage of the girth maximum BTU search process directly flows from the micro-partition cycle maximization criterion and is a logical extension of the results proved in  $[2]$.
\subsection {Scaling of a permutation}
{Scaling of a partition has been defined in $[2]$ . We similarly define scaling of a permutation.} 
\begin{definition}
Given a permutation $q\in S_{b\ast k}$ , the scaling of permutation is defined as
 $s:S_{b\ast k}\to S_{b\ast k^{2}}$ where  $q\in S_{b\ast k}$ is mapped to  $p\in S_{b\ast k^{2}}$ in the following manner \\
for( $j=1;j\le b\ast k;j$++) \{ \\
\ \  $p$.[depth  $j$] = $q$.[depth  $j_{1}$] +  $j_{2}\ast k$; \\
\} \\
 $j_{1}=j$ mod  $k$, \ is the depth of  $q$ in  $S_{b\ast k}$ and $j_{2}$ is the integer part of  $(j/k)$.
 $s$ maps  $S_{b\ast k}$ to  $\{(b\ast k)!\}^{k}$ elements of  $S_{b\ast k^{2}}$ which is a significantly smaller set than  $S_{b\ast k^{2}}$ which has $(b\ast k^{2})!$ elements.
\end{definition}

\subsection {Notation for scaling}
{We denote the scaled permutation $p\in S_{b\ast k^{2}}$ as  $k\ast q$ where  $q\in S_{b\ast k}$ .}
\begin{figure}[htbp]
\begin{equation*}
\left[\begin{matrix}0&0&1&0\\0&0&0&1\\1&0&0&0\\0&1&0&0\end{matrix}\right]\rightarrow \left[\begin{matrix}0&0&1&0&0&0&0&0\\0&0&0&1&0&0&0&0\\1&0&0&0&0&0&0&0\\0&1&0&0&0&0&0&0\\0&0&0&0&0&0&1&0\\0&0&0&0&0&0&0&1\\0&0&0&0&1&0&0&0\\0&0&0&0&0&1&0&0\end{matrix}\right]
\end{equation*}
\begin{equation*}
(\begin{matrix}3&4&1&2\end{matrix})\rightarrow (\begin{matrix}3&4&1&2&7&8&5&6\end{matrix})
\end{equation*}
\caption{Example for Scaling of a permutation in $S_{4}$ to $S_{8}$ with matrix and permutation representations}
\end{figure}

\subsection {Explanation of micro-partition cycle maximization}
Explanation of micro-partition cycle maximization has been provided in $[2]$ and is one of the key arguments that allows us to apply restrictions to the Search Space in order to find a girth maximum $(m, r$ BTU.

\begin{theorem}  {Scaling Theorem for a girth maximum $(m,r)$ BTU for $r \ge 3$} \\
There exists a girth maximum $(m,r)$  BTU in  $Z(m, r)$.
\end{theorem}
\begin{proof}  The proof has been provided in the Appendix.
\end{proof}

\begin{corollary} {Scaling Theorem for a girth maximum $(m,3)$ BTU} \\
There exists a girth maximum $(m,3)$ BTU in $Z(m, 3)$.
\end{corollary}

\begin{theorem} 
Every non-isomorphic $(m, r)$ in $Z(m, r)$ can be generated by the following algorithm \\
We find $b,k\in \mathbb{N}$ such that $b$ is the smallest integer satisfying $m=b\ast k^{r-1}$; \\
for(  $i=2;i<r;i$++) \{ \\
\  $p_{i}=C_{j};\mathit{min}(b\ast k^{i-1}-j,j)>b\ast k^{i-2}$ such that $(j,b\ast k^{i-1},b\ast k^{i-1}-j)$ are relatively prime; \\
if( $i$ \ \ == \  $2$ \ ) \\
$p_{i-1}=I_{b\ast k^{i-1}};$  \\
\ else \{ \\
\ Rearrange the $(b\ast k^{i-1},i)$ BTU such that \ \  $p_{i-1}=I_{b\ast k^{i-1}}$; \\
We enumerate all $q_{i-2}\in S_{b\ast k^{i-2}}$ such that a $(b\ast k^{i-1},i)$ BTU is formed by  $p_{1},\ldots ,p_{i}\in S_{b\ast k^{i-1}};p_{i -2}=k\ast q_{i -2}$; \\
if(i != r -- 1) \\
Scale permutations $p_{y}=k\ast q_{y};1\le y\le i$; \\
\} \\
\}
\end{theorem}
\begin{proof} Consider an arbitrary element of $A \in Z(m, r)$ with permutations \{${p_{1} , p_{2} , \ldots, p_{r}} \in S_{m}$\} such that
\begin{enumerate}
\item $\beta _{i}$ refers to the optimal partition given by $\sum _{j=1}^{r-1-i}b\ast k^{i}=b\ast k^{r-1}$ between permutations $p_{i}$  and  $p_{i + 1}$ for $1\le i\le r-1$.
\item $p_{r - 1} =I_{m}$, where $I_m$ is the identity permutation on a set of $m$ elements.
\item $p_{j}=k^{r -1 - j}\ast q_{j}$ for $1 \le j \le r - 2$, and $q_{j} \in S_{b \ast k^{j}}$.
\end{enumerate}
It is clear that $Z(m, r)$ consists of enumerations of \{$q_{1}\in S_{b\ast k^{1}}, q_{2}\in S_{b\ast k^{2}}, \ldots, q_{r - 2}\in S_{b\ast k^{r - 2}}$ \} and this is precisely the same set of enumerations given by the proposed algorithm.
\end{proof}

\begin{theorem} 
Given a girth maximum $(k^{i-1}, i)$ BTU with permutations $p_{1},\ldots ,p_{i}\in S_{k^{i-1}}$, with $\beta _{l}$ referring to the optimal partition given by $\sum _{j=1}^{r-1-l}k^{l}=k^{r-1}$ between permutations $p_{l}$  and  $p_{l + 1}$ for $1\le l\le i-1$. If we scale all permutations by $k$ and find permutation $p_{i + 1} \in S_{k^{i}}$ such that it leads to maximum girth, then the resultant $(k^{i}, i + 1)$ BTU is a girth maximum $(k^{i}, i + 1)$ BTU for all integers $i \ge 2$.
\end{theorem}
\begin{proof}  The proof has been provided in the Appendix.
\end{proof}

\begin{lemma} 
Given $m, r \in \mathbb{N}$ such that $r < m/2$, let $b,k\in \mathbb{N}$ satisfying $m=b\ast k^{r-1}$ such that $b$ is minimized, and if the prime factorization of the greatest common divisor of $b$ and $k$ do not have a non-trivial power of a prime, the following statement is true.
Given a girth maximum $(k^{i-1}, i)$ BTU with permutations $p_{1},\ldots ,p_{i}\in S_{k^{i-1}}$, with $\beta _{l}$ referring to the optimal partition given by $\sum _{j=1}^{r-1-l}k^{l}=k^{r-1}$ between permutations $p_{l}$  and  $p_{l + 1}$ for $1\le l\le i-1$. If we scale all permutations by $k$ and find permutation $p_{i + 1} \in S_{k^{i}}$ such that it leads to maximum girth, then the resultant $(k^{i}, i + 1)$ BTU is a girth maximum $(k^{i}, i + 1)$ BTU for all integers $i \ge 2$.
\end{lemma}
\begin{proof}  The proof has been provided in the Appendix.
\end{proof}

\section{High Level Description Of Enumeration Based Search Algorithm for Girth Maximum $(m, r)$ BTU in $Z(m, r)$}
The Enumeration Based Search Algorithm for Girth Maximum $(m, r)$ BTU is derived from the algorithm to enumerate elements of $Z(m, r)$ described in Theorem $4$.
\subsection{Enumeration Based Search algorithm for girth maximum $(m, r)$ BTU for\ $r>3$} 
We find $b,k\in \mathbb{N}$ such that $b$ is the smallest integer satisfying $m=b\ast k^{r-1}$; \\
for(  $i=2;i<r;i$++) \{ \\
\  $p_{i}=C_{j};\mathit{min}(b\ast k^{i-1}-j,j)>b\ast k^{i-2}$ such that $(j,b\ast k^{i-1},b\ast k^{i-1}-j)$ are relatively prime; \\
if( $i$ \ \ == \  $2$ \ ) \\
$p_{i-1}=I_{b\ast k^{i-1}};$  \\
\ else \{ \\
\ Rearrange the $(b\ast k^{i-1},i)$ BTU such that \ \  $p_{i-1}=I_{b\ast k^{i-1}}$; \\
Find $q_{i-2}\in S_{b\ast k^{i-2}}$ such that it maximizes girth of $(b\ast k^{i-1},i)$ BTU is formed by  $p_{1},\ldots ,p_{i-2}\in S_{b\ast k^{i-1}};p_{x}=k\ast q_{x};1\le x\le i-2$; \\
if(i != r -- 1) \\
Scale permutations $p_{y}=k\ast q_{y};1\le y\le i$; \\
\} \\
\}

\subsection{Enumeration Based Search algorithm for a Girth Maximum $(m, 3)$ BTU  where $ m = b \ast k^{2}$}
We find $b,k\in \mathbb{N}$ such that $b$ is the smallest integer satisfying $m=b\ast k^{2}$; \\
for(  $i=2;i<3;i$++) \{ \\
\  $p_{i}=C_{j};\mathit{min}(b\ast k^{i-1}-j,j)>b\ast k^{i-2}$ such that $(j,b\ast k^{i-1},b\ast k^{i-1}-j)$ are relatively prime; \\
if( $i$ \ \ == \  $2$ \ ) \\
$p_{i-1}=I_{b\ast k^{i-1}};$  \\
\ else \{ \\
\ Rearrange the $(b\ast k^{i-1},i)$ BTU such that \ \  $p_{i-1}=I_{b\ast k^{i-1}}$; \\
Find $q_{1} \in S_{b\ast k}$ such that a girth maximum  $(b\ast k^{2},3)$ BTU is formed by  $p_{1}, p_{2}, p_{3}\in S_{b\ast k^{i-1}};p_{1}=k\ast q_{1}$; \\
\}

\section {Search Space}
\begin{theorem}
If the first permutation is randomly chosen from a $m$  Symmetric Permutation Tree, the number of unique ways to choose a
second permutation that has a partition of $(m)\in P_{2}(m)$ with the first permutation is  $(m-1)!$ .
\end{theorem} 
\begin{proof} This directly follows from the permutation enumeration formula proved in  $[1]$ . 
\end{proof}

\begin{theorem} 
The permutations having a partition of $(m)\in P_{2}(m)$ with a given permutation on a \  $m$  Symmetric Permutation Tree can be represented by a Symmetric Group of degree  $m\text{--}1$ .
\end{theorem} 
\begin{proof} The proof has been provided in the Appendix.
\end{proof}

\begin{corollary} 
The Enumeration Based Search Algorithm for a girth maximum $(m, 3)$ BTU in $Z(m, 3)$ is in EXPTIME. 
\end{corollary} 

\subsection {Notation}
Given a permutation  $p\in S_{m}$ , let us denote the set of permutations  $q\in S_{m}$  such that  $q\notin C(p)$ and the partition between $q$  and  $p$  is $(m)\in P_{2}(m)$ as $A_{(m)}(p)$ .

\begin{definition} {Remapping function} \\
We define a remapping label function  $f_{R}:A_{(m)}(p)\to S_{m-1}$ that remaps node labels for  $A_{(m)}(p)\subset S_{m}$ Given permutation
$(x_{1},x_{2},\ldots ,x_{m})\in A_{(m)}(p)$ \\
$\begin{gathered}f_{R}:A_{(m)}(p)\to S_{m-1}\\(x_{1}x_{2}\ldots x_{i-1}\ldots
x_{m})\to (y_{1}y_{2}\ldots y_{i}\ldots y_{m-1})\end{gathered}$ \\
for ( $i=1$  ; $i<m$ ;  $i$ ++) \{ \\
Remapping of node labels at depth  $i$ \textbf{: }
\begin{enumerate}
\item[] {\bfseries
{\ } $\begin{gathered}j\to j,\mathit{if}j<x_{i}\\j\to
j\text{--}1,\mathit{if}j\ge x_{i}\end{gathered}$ {;}} \\
\} \\
We remove all leaf nodes at depth  $m$  .
\end{enumerate}
\end{definition}

\begin{theorem}
{The inverse function of the Remapping label function $[f_{R}]$ exists .}
\end{theorem}
\begin{proof}  The proof has been provided in the Appendix.
\end{proof}

\begin{corollary} 
The remapping function $[f_{R}]$ maps each element of $A_{(m)}(p)$ to an element of $S_{m\text{--}1}$ .
\end{corollary} 

\begin{corollary} 
The inverse remapping function  $[f_{R}]^{-1}$ maps each element of  $S_{m\text{--}1}$ to an element of  $A_{(m)}(p)$ ,
where  $p$  is the first chosen permutation from \ a  $m$  Symmetric Permutation Tree.
\end{corollary}

\section {Search Space: Cayley Graph of A Symmetric Group}
{The search for a girth maximum $(m,r)$ BTU can be broken into $r\text{--}2$  steps for  $r \ge 3$ where each $q_{i}\in S_{b\ast k^{i}}$  is scaled, 
$p_{i}=k^{r\text{--}1\text{--}i}\ast q_{i}$ ; with each $p_{i}\in S_{b\ast {r-1}}$ for  $1\le i\le r-2$. }
{Without loss of generality, at each step of the search process, we can transform the $(b\ast k^{i-1},i)$ BTU with 
$p_{i}=I_{b\ast k^{i-1}}$ such that  $\{q_{i}\notin C(I_{b\ast k^{i}})/q_{i}\in S_{b\ast k^{i}}\}=A_{(b\ast k^{i})}(I_{b\ast k^{i}})$.}

\begin{theorem}
The search space for each permutation $q_{i}\in S_{b\ast k^{i}}$ {where } $p_{i}=k^{r\text{--}1\text{--}i}\ast q_{i};p_{i}\in S_{b\ast {r-1}}$
{ for } $1\le i\le r-2$ {in a }girth maximum $(m,r)$ BTU can be mapped to a Cayley Graph of a Symmetric Group of degree  $b\ast k^{i}\text{--}1$ where
\begin{enumerate}
\item {
 $b,k\in \mathbb{N}$  satisfy $m=b\ast k^{r-1}$ such that  $b$ is the smallest integer satisfying the equation .}
\item {
 $\beta _{1},\beta _{2},\ldots ,\beta _{r-1}\in P_{2}(b\ast k^{r-1})$ refer to optimal partitions derived in $[2]$ such that the girth
maximum $(m,r)$  BTU \ lies in  $\Phi (\beta _{1},\beta _{2},\ldots,\beta _{r-1})$ ,where  $\beta _{i}$  refers to  $\sum
_{j=1}^{r-1-i}b\ast k^{i}=b\ast k^{r-1}$ for  $1\le i\le r-1$ . Thus, $\beta _{1},\beta _{2},\ldots ,\beta _{r-1}\in P_{2}(b\ast k^{r-1})$
are  $\sum _{j=1}^{k^{r-2}}b\ast k=b\ast k^{r-1}$ ,  $\sum_{j=1}^{k^{r-3}}b\ast k^{2}=b\ast k^{r-1}$ ,  $\ldots $ ,  $\sum_{j=1}^{k}b\ast k^{r-2}=b\ast k^{r-1}$ and  $\sum _{j=1}^{1}b\ast
k^{r-1}=b\ast k^{r-1}$ respectively.}
\end{enumerate}
\end{theorem}
\begin{proof}  The proof has been provided in the Appendix.
\end{proof}

\begin{corollary} 
The search space for a girth maximum $(m,3)$ BTU can be mapped to a Cayley Graph of a Symmetric Group of degree  $b\ast k\text{--}1$ where
\begin{enumerate}
\item  $b,k\in \mathbb{N}$  satisfy  $m=b\ast k^{2}$ such that  $b$ is the smallest integer satisfying the equation .
\item $\beta _{1},\beta _{2}\in P_{2}(b\ast k^{2})$ refer to optimal partitions derived in $[2]$ such that the girth maximum  $(m,3)$  BTU lies in  $\Phi (\beta _{1},\beta _{2})$ , i.e.,  $\sum
_{j=1}^{k}b\ast k=b\ast k^{2}$ and  $\sum _{j=1}^{1}b\ast k^{2}=b\ast k^{2}$ respectively.
\end{enumerate}
\end{corollary} 

\section {Interpretation of this search space.}
 
\begin{figure}[htbp]
\includegraphics[scale=0.22]{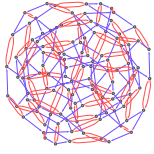}
\caption{Cayley Graph of degree $5$ obtained as a screenshot from Mathematica Software}
\end{figure}

\subsection {Permutation and Associated Partition for each node of the Cayley Graph of the Symmetric Group of degree  $b\ast k^{i}\text{--}1$ }
For each node of the Cayley Graph of the Symmetric Group of degree $b\ast k^{i}\text{--}1$ , we associate a permutation $q_{i-2}$ and an associated
partition between $p_{i}$  and $p_{i-2}$ , since the search process involves optimizing the partition between  $p_{i}$  and  $p_{i-2}$ at each step.

\section{Conclusion }
This paper describes an algorithm for finding a girth maximum $(m,r)$ BTU by building upon the theoretical results developed earlier in $[1]$ and  $[2]$ and also describes the structure of the search space for each permutation. We also show that the search algorithm for each permutation that maximizes girth is in EXPTIME. The detailed implementation and its comparison from other known results from the available literature have been discussed in $[5]$.

\newpage
\section {Appendix}

\subsection {Proof for Theorem $1$}
\begin{proof} Let $T$ be a labeled  $(m,r)$  BTU in  $\Phi (\beta _{1},\beta_{2},\ldots ,\beta _{r\text{--}1})$ represented by a set of compatible permutations ${p_{1} , p_{2} , \ldots, p_{r}}$. 
Let us transform $T$ by row and column permutations to a labeled $(m,r)$  BTU, $T_1$ which now has $p_{r - 1} =I_{m}$. Since $T_1$ is isomorphic to $T$ and $T_1 \in X(m, r)$, every labeled  $(m,r)$  BTU in  $\Phi (\beta _{1},\beta_{2},\ldots ,\beta _{r\text{--}1})$ is isomorphic to some labeled $(m,r)$  BTU in  $X(m, r)$.
\end{proof}

\subsection {Proof for Theorem $2$}
\begin{proof}
Let us assume that $k, 1 \le k < m$ and $m$ are not relatively prime.  Let $l = min(k, m - k)$. Since $l$ divides $m$, we have $m/l$ cycles each of length $2 * l$ and hence the partition between $I_m$ and $C_j$ is $\sum _{i=1}^{l}(m/l) = m$. It follows that if $k$ and $m$ are relatively prime the partition is $(m) \in P_2(m)$ since we have only one cycle of length $2 \ast m$. 
\end{proof}

\subsection {Proof for Theorem $3$}
\begin{proof}
The derivation of optimal partitions that lead to a girth maximum BTU in $[2]$ refer to the partitions between permutations $p_{i}\in S_{b\ast k^{r-1}}$ and  $p_{i+1}\in S_{b\ast k^{r-1}}$ for  $1\le i\le r-1$. Let us consider the other partitions and consider maximization of the micro-partition cycles for each of them. We now consider the partitions between permutations  $p_{u}\in S_{b\ast k^{2}}$ and $p_{v}\in S_{b\ast k^{2}}$ where  $1\le u\le r-1;1\le v\le r-1;u>v+1$ and apply the micro-partition cycle maximization criterion to obtain $p_{i}=k^{r\text{--}1\text{--}i}\ast q_{i}$ where $q_{i}\in S_{b\ast k^{i}}$ for $1\le i\le r-1$ for the girth maximum BTU. Thus, the proof here is for the existence of a girth maximum $(m,r)$ 
BTU with the form $p_{i}=k^{r\text{--}1\text{--}i}\ast q_{i}$ where $q_{i}\in S_{b\ast k^{i}}$ for $1\le i\le r-1$ .
\end{proof}

\subsection {Proof for Theorem $5$}
\begin{proof} Let us prove the statement by mathematical induction. For $i =2$, a girth maximum $(k, 2)$ BTU is scaled by a factor of $k$ and a permutation is found with partition $(k^2) \in P_2(k^2)$ that maximizes the girth of the $(k^2, 3)$ BTU. Hence, the statement is true for $i = 2$. Let us assume that the statement is true for $i=z$. We need to establish that the statement for $i=z + 1$. Let $A$ be the girth maximum $(k^{z-1}, z)$ BTU $\in P_2(k^{z-1})$. Let us scale all the permutations in $A$ by a factor of $k$ and add a permutation with partition $(k^z) \in P_2(k^z)$ that maximizes the girth of the consequent $(k^{z}, z+ 1)$ BTU that we now refer to as $B$. We need to establish that $B$ is also a girth maximum $(k^{z}, z+ 1)$ BTU. Let us consider the partitions \{ $p_{1} , p_{2} , \ldots, p_{z},  p_{z + 1}$ \} such that $p_{1} , p_{2}$ are scaled versions of the original girth maximum $(k, 2)$ BTU, $p_{1} , p_{2}, p_{z}$ are scaled versions of the original girth maximum $(k^2, 3)$ BTU and $p_{1} , p_{2} , \ldots, p_{z}$ are scaled versions of the original girth maximum $(k^{z - 1}, z)$ BTU. It is now clear that the partition between $p_{1}$ and $p_{2}$ is $\sum _{j=1}^{z-1}k^{1}=k^{z}$, partition between $p_{2}$ and $p_{3}$ is $\sum _{j=1}^{z-2} k^{2}= k^{z}$, and that the partition between $p_{z}$ and $p_{z + 1}$ is $\sum _{j=1}^{1} k^{z}= k^{z}$ which correspond to the optimal partitions that lead to a girth maximum BTU. Further, at each of the process, partitions between other pairs of permutations $p_{x}$ and $p_{y}$ where $ 1 \le y < x + 1 \le r$ are also optimized at each stage of the BTU building process, and the scaled versions of the permutations continue to preserve their optimality. 
Since $p_{z + 1}$ maximizes the girth of the $(k^{z}, z+ 1)$ BTU, and since the partitions between other pairs of permutations are also optimized, we conclude that this is indeed a girth maximum  exactly $(k^{z}, z+ 1)$ BTU, thus proving the statement for all $i >2$.
\end{proof}

\subsection {Proof for Lemma $1$}
\begin{proof} If the prime factorization of the greatest common divisor of $b$ and $k$ have a non-trivial power of a prime, $h^t , h,t\in \mathbb{N}$, then for the case of a $(b * k^{t}, t + 1)$ BTU, the girth maximum BTU the parameters $b$ and $k$ get replaced by $b/h^{t}$ and $k\ast h$ for which the statement is clearly not valid, since $b/h^{t}$ minimized with $ b * k^t = (b/h^t) \ast k\ast h$. If the prime factorization of the greatest common divisor of $b$ and $k$ does have a non-trivial power of a prime, then Theorem $5$ applies and the statement is valid.
\end{proof}

\subsection {Proof for Theorem $7$}
\begin{proof} Since the number of permutations having a partition of $(m)\in P_{2}(m)$ with a given permutation on a \  $m$  Symmetric
Permutation Tree is  $(m-1)!$ , we can map each of them to Symmetric Group of degree  $m\text{--}1$ which has a known order of  $(m-1)!$ .
\end{proof}

\subsection {Proof for Theorem $8$}
\begin{proof} 
Given first permutation \ $p=(t_{1},t_{2},\ldots,t_{m})\in S_{m}$ for  $A_{(m)}(p)$ , we define a function  $h:S_{m-1} \to A_{(m)}(p)$, 
\\ $(y_{1} y_{2} \ldots y_{i}\ldots y_{m-1})\to (x_{1} x_{2}\ldots x_{i-1}\ldots x_{m})$  \\
for( $i=1$  ; $i\le m$ ;  $i$ ++) \{ \\
{Depth } $i$  :  
 $\begin{gathered}j\to j\ \forall j<t_{i}\\j\to j+1\ \forall j\ge
t_{i}\end{gathered}$ . \\
\} \\
We can verify that  $h$  is indeed \  $[f_{R}]^{-1}$ .
\end{proof}

\subsection {Proof for Theorem $9$}
\begin{proof} 
{We have $min(r-2, 0)$ steps of search for a girth maximum $(m,r)$ BTU.}
At each step  $i;1\le i\le r-2$  of the search process we re-arrange the
BTU and map  $p_{i}$  to  $I_{b\ast k^{i}}$ . The permutations get
progressively scaled by  $k$  after addition of the second permutation,
 $\ldots $ ,  $(r\text{--}1)^{\mathit{th}}$ permutation after different
steps of the \ of the search process.
In each step, we search for  $q_{i}\in S_{b\ast k^{i}}$ that leads to the
best girth and has a partition  $(b\ast k^{i})\in P_{b\ast k^{i}}$ 
with  $p_{i}=I_{b\ast k^{i}}$ and hence can be mapped to the set 
$A_{(b\ast k^{i})}(I_{b\ast k^{i}})$ . We have proved that the set 
$A_{(b\ast k^{i})}(I_{b\ast k^{i}})$ is isomorphic to  $S_{b\ast
k^{i}\text{--}1}$ with mapping from  $A_{(b\ast k^{i})}(I_{b\ast
k^{i}})$ to  $S_{b\ast k^{i}\text{--}1}$ with the label remapping
function  $[f_{R}]$ , and the mapping from  $S_{b\ast k^{i}\text{--}1}$
to  $A_{(b\ast k^{i})}(I_{b\ast k^{i}})$ with the inverse label
remapping function  $[f_{R}]^{-1}$ . Since the search space for 
$q_{i}\in S_{b\ast k^{i}}$ is now mapped to a Symmetric Group of degree
 $b\ast k^{i}\text{--}1$ which has an order  $(b\ast k^{i}\text{--}1)!$
, we consider the Cayley Graph of the Symmetric Group of degree  $b\ast
k^{i}\text{--}1$ .
Each node of a Cayley Graph of the Symmetric Group of degree  $b\ast
k^{i}\text{--}1$ is connected to other  $b\ast k^{i}-2$  nodes, and any
other node can be reached in  $b\ast k^{i}$ optimal node transitions.
\end{proof}

\end{document}